\documentclass[12pt]{article}
\usepackage{amsthm}
\usepackage{amsfonts}
\usepackage{epsfig}
\usepackage{amssymb}
\usepackage{amsmath}
\usepackage{amscd}
\usepackage{colonequals}
\usepackage{url}
\newtheorem{theorem}{Theorem}
\newtheorem{lemma}[theorem]{Lemma}
\newtheorem{corollary}[theorem]{Corollary}
\newtheorem{question}[theorem]{Question}
\newtheorem{observation}[theorem]{Observation}

\newcommand{\Aa}{\mathcal{A}}

\newcommand{\GG}{\mathcal{G}}

\title{Additive non-approximability of chromatic number in proper minor-closed classes}

\author{Zden\v{e}k Dvo\v{r}\'{a}k\thanks{Charles University, Prague, Czech Republic.
E-mail: {\tt rakdver@iuuk.mff.cuni.cz}.  Supported by project 17-04611S (Ramsey-like aspects of graph
coloring) of Czech Science Foundation.}\and
Ken-ichi Kawarabayashi\thanks{National Institute of Informatics,
2-1-2 Hitotsubashi, Chiyoda-ku, Tokyo 101-8430, Japan. Supported by JST ERATO Grant Number JPMJER1305, Japan.}}
\date{}
\begin{document}
\maketitle

\begin{abstract}
Robin Thomas asked whether for every proper minor-closed class $\GG$, there exists a polynomial-time
algorithm approximating the chromatic number of graphs from $\GG$ up to a constant additive error independent
on the class $\GG$.  We show this is not the case: unless $\text{P}=\text{NP}$, for every integer $k\ge 1$, there is no polynomial-time
algorithm to color a $K_{4k+1}$-minor-free graph $G$ using at most $\chi(G)+k-1$ colors.
More generally, for every $k\ge 1$ and $1\le\beta\le 4/3$, there is no polynomial-time
algorithm to color a $K_{4k+1}$-minor-free graph $G$ using less than $\beta\chi(G)+(4-3\beta)k$
colors.  As far as we know, this is the first non-trivial non-approximability result regarding the chromatic number in proper minor-closed
classes.

We also give somewhat weaker non-approximability bound for $K_{4k+1}$-minor-free graphs with no cliques of size $4$.
On the positive side, we present additive approximation algorithm whose error depends on the apex number of
the forbidden minor, and an algorithm with additive error 6 under the additional assumption that the graph has no $4$-cycles.
\end{abstract}

The problem of determining the chromatic number, or even of just deciding whether a graph is colorable
using a fixed number $c\ge 3$ of colors, is NP-complete~\cite{garey1979computers}, and thus it cannot be
solved in polynomial time unless $\text{P}=\text{NP}$.  Even the approximation version of the problem is hard:
for every $\varepsilon>0$, Zuckerman~\cite{colnonap} proved that unless $\text{P}=\text{NP}$, there exists no polynomial-time algorithm approximating
the chromatic number of an $n$-vertex graph within multiplicative factor $n^{1-\varepsilon}$.

There are more restricted settings in which the graph coloring problem becomes more tractable.
For example, the well-known Four Color Theorem implies that deciding $c$-colorability of a planar graph
is trivial for any $c\ge 4$; still, $3$-colorability of planar graphs is NP-complete~\cite{garey1979computers}.
From the approximation perspective, this implies that chromatic number of planar graphs can be approximated in polynomial time
up to multiplicative factor of $4/3$ (but not better), and additively up to $1$.

More generally, the result of Thomassen~\cite{Thomassen93} on 6-critical graphs implies that the $c$-coloring problem
restricted to graphs that can be drawn in any fixed surface of positive genus is polynomial-time solvable for any $c\ge 5$.
The case $c=3$ includes $3$-colorability of planar graphs and consequently is NP-complete, while the complexity of $4$-colorability
of embedded graphs is unknown for all surfaces of positive genus.  Consequently, chromatic number of graphs embedded in
a fixed surface can be approximated up to multiplicative factor of $5/3$ and additively up to $2$.

If a graph can be drawn in a given surface, all its minors can be drawn there as well.  Hence, it is natural to also consider
the coloring problem in the more general setting of proper minor-closed classes.  Further motivation for this setting comes
from Hadwiger's conjecture, stating that all $K_k$-minor-free graphs are $(k-1)$-colorable.
This conjecture is open for all $k\ge 7$, and not even a polynomial-time algorithm to decide $(k-1)$-colorability
of $K_k$-minor-free graphs is known (Kawarabayashi and Reed~\cite{decihad} designed an algorithm that for a given
input graph $G$ finds a $(k-1)$-coloring, or a minor of $K_k$ in $G$, or finds a counterexample to Hadwiger's conjecture).
However, any $K_k$-minor-free graph is $O(k\sqrt{\log k})$-colorable~\cite{kostomindeg}.  This implies that for every proper
minor-closed class $\GG$, if $k$ is the minimum integer such that $K_k\not\in \GG$, then there exists a constant $c=O(k\sqrt{\log k})$
such that every graph in $\GG$ is $c$-colorable, and thus chromatic number of graphs in $\GG$ can be approximated up to
multiplicative factor $c/3$ and additively up to $c-3$.

On the hardness side, consider for any planar graph $G$ and an integer $t\ge 0$ the graph $G_t$ obtained from $G$ by adding
$t$ universal vertices (adjacent to every other vertex of $G_t$).  Then $\chi(G_t)=\chi(G)+t$, and since $3$-colorability
of planar graphs is NP-complete, there cannot exist a polynomial-time algorithm to decide whether such a graph $G_t$ is $(t+3)$-colorable,
unless $\text{P}=\text{NP}$.  Furthermore, $G_t$ does not contain $K_{t+5}$ as a minor; indeed, each minor of $G_t$ has an induced planar
subgraph containing all but $t$ of its vertices, which is not the case for $K_{t+5}$.  This outlines the importance of another
graph parameter in this context, the \emph{apex number}: we say a graph $H$ is \emph{$t$-apex} if there exists a set $X$ of vertices of $H$ of size
at most $t$ such that $H-X$ is planar, and the apex number of $H$ is the minimum $t$ such that $H$ is $t$-apex. 
The presented construction shows that if the apex number of $H$ is at least $t$, then $(t+2)$-colorability is
NP-complete even when restricted to the class of $H$-minor-free graphs.
On the positive side, Dvo\v{r}\'ak and Thomas~\cite{apex} gave, for any $t$-apex graph $H$ and integer $c\ge t+4$, a polynomial-time algorithm
to decide whether a $(t+3)$-connected $H$-minor-free graph is $c$-colorable (the connectivity assumption is necessary, since
they also proved that for every integer $t\ge 1$, there exists a $t$-apex graph $H$ such that testing $(t+4)$-colorability of
$(t+2)$-connected $H$-minor-free graphs is NP-complete).

In all the mentioned results for proper minor-closed classes, the number of colors needed and thus also the
magnitude of error of the corresponding approximation algorithms depended on the specific class.  This contrasts with the case
of embedded graphs: the multiplicative ($5/3$) and additive ($2$) errors of these approximation algorithms
are independent on the fixed surface in that the graphs are drawn.
Hence, it is natural to ask the following questions.
\begin{question}\label{que-multi}
Does there exist $\beta\ge 1$ with the following property: for every proper minor-closed class $\GG$,
there exists a polynomial-time algorithm taking as an input a graph $G\in\GG$
and returning an integer $c$ such that $\chi(G)\le c\le \beta\chi(G)$?
\end{question}

\begin{question}\label{que-addit}
Does there exist $\alpha\ge 0$ with the following property: for every proper minor-closed class $\GG$,
there exists a polynomial-time algorithm taking as an input a graph $G\in\GG$
and returning an integer $c$ such that $\chi(G)\le c\le \chi(G)+\alpha$?
\end{question}

That is, is it possible to approximate chromatic number up to a multiplicative or additive error independent
on the considered class of graphs $\GG$, as long as $\GG$ is proper minor-closed?  Perhaps a bit surprisingly, the answer to Question~\ref{que-multi}
is positive.  As shown by DeVos et al.~\cite{devospart} and algorithmically by Demaine et al.~\cite{demaine2005algorithmic},
for every proper minor-closed class $\GG$, there exists a constant $\gamma_{\GG}$ such that the vertex set
of any graph $G\in \GG$ can be partitioned in polynomial time into two parts $A$ and $B$ with both
$G[A]$ and $G[B]$ having tree-width at most $\gamma_{\GG}$.  Consequently, $\chi(G[A]),\chi(G[B])\le \chi(G)$ can be determined
exactly in linear time~\cite{coltweasy}, and we can color $G[A]$ and $G[B]$ using disjoint sets of colors, obtaining a coloring of $G$ using
at most $2\chi(G)$ colors.  That is, $\beta=2$ has the property described in Question~\ref{que-multi}.

In the light of this result, Question~\ref{que-addit} may seem more tractable.  Thomas~\cite{robinconj} conjectured
that such a constant $\alpha$ exists, and Kawarabayashi et al.~\cite{kawarabayashi2009additive} conjectured that this is the case
even for list coloring.  As our main result, we disprove these conjectures.
\begin{theorem}\label{thm-nonapprox}
Let $k_0$ be a positive integer, let $F$ be a $(4k_0-3)$-connected graph that is not $(4k_0-4)$-apex,
and let $1\le \beta\le 4/3$ be a real number.
Unless $\text{P}=\text{NP}$, there is no polynomial-time
algorithm taking as an input an $F$-minor-free graph $G$
and returning an integer $c$ such that $\chi(G)\le c<\beta\chi(G)+(4-3\beta)k_0$.
\end{theorem}
In particular, in the special case of $\beta=1$ and $F$ being a clique, we obtain the following.

\begin{corollary}\label{cor-nonapprox}
Let $k_0$ be a positive integer.
Unless $\text{P}=\text{NP}$, there is no polynomial-time
algorithm taking as an input a $K_{4k_0+1}$-minor-free graph $G$
and returning an integer $c$ such that $\chi(G)\le c\le \chi(G)+k_0-1$.
\end{corollary}

On the positive side, Kawarabayashi et al.~\cite{kawarabayashi2009additive} showed it is possible to approximate chromatic
number of $K_k$-minor free graphs in polynomial time additively up to $k-2$.  We leave open the question whether 
a better additive approximation (of course above the bound $\approx k/4$ given by Corollary~\ref{cor-nonapprox}) is possible.

Another positive result was given by Demaine et al.~\cite{demaine2009approximation}, who proved that if $H$ is a $1$-apex
graph, then the chromatic number of $H$-minor-free graphs can be approximated additively up to $2$.  Let us also
remark that if $H$ is $0$-apex (i.e., planar), then $H$-minor-free graphs have bounded tree-width~\cite{RSey},
and thus their chromatic number can be determined exactly in linear time~\cite{coltweasy}.
We generalize these results to excluded minors with larger apex number (the relevance of the apex number
in the context is already showcased by Theorem~\ref{thm-nonapprox}).

\begin{theorem}\label{thm-alg}
Let $t$ be a positive integer and let $H$ be a $t$-apex graph.
There exists a polynomial-time algorithm taking as an input an $H$-minor-free graph $G$
and returning an integer $c$ such that $\chi(G)\le c\le \chi(G)+t+3$.
\end{theorem}

The construction we use to establish Theorem~\ref{thm-nonapprox} results in graphs with large clique number (on the order of $k_0$).
On the other hand, forbidding triangles makes the coloring problem for embedded graphs more tractable---all planar graphs are $3$-colorable~\cite{grotzsch1959}
and there exists a linear-time algorithm to decide $3$-colorability of a graph embedded in any fixed surface~\cite{trfree7}.
It is natural to ask whether Question~\ref{que-addit} could not have a positive answer for triangle-free graphs, and this question
is still open.  On the negative side, we show that forbidding cliques of size $4$ is not sufficient.

\begin{theorem}\label{thm-k4free}
Let $\beta$ and $d$ be real numbers such that $1\le \beta<4/3$ and $d\ge 0$.  Let $m=\lceil d/(4-3\beta)\rceil$.
There exists a positive integer
$k_0=O\bigl(m^4\log ^2 m\bigr)$ such that the following holds.  Let $F$ be a $(4k_0-1)$-connected graph with at least $4k_0+8$ vertices
that is not $(4k_0-4)$-apex.  Unless $\text{P}=\text{NP}$, there is no polynomial-time
algorithm taking as an input an $F$-minor-free graph $G$ with $\omega(G)\le 3$
and returning an integer $c$ such that $\chi(G)\le c<\beta\chi(G)+d$.
\end{theorem}

In particular, in the special case of $\beta=1$ and $F$ being a complete graph, we get the following.

\begin{corollary}
For every positive integer $k_0$, there exists an integer $d=\Omega(k_0^{1/4}/\log^{1/2}k_0)$ as follows.
Unless $\text{P}=\text{NP}$, there is no polynomial-time
algorithm taking as an input a $K_{4k_0+8}$-minor-free graph $G$ with $\omega(G)\le 3$
and returning an integer $c$ such that $\chi(G)\le c\le \chi(G)+d$.
\end{corollary}

On the positive side, we offer the following small improvement to the additive error of Theorem~\ref{thm-alg}.

\begin{theorem}\label{thm-algnotri}
Let $t$ be a positive integer and let $H$ be a $t$-apex graph.
There exists a polynomial-time algorithm taking as an input an $H$-minor-free graph $G$ with no triangles
and returning an integer $c$ such that $\chi(G)\le c\le \chi(G)+\lceil (13t+172)/14\rceil$.
\end{theorem}

What about graphs of larger girth?  It turns out that Question~\ref{que-addit} has positive answer for graphs of
girth at least $5$, with $\alpha=6$.  Somewhat surprisingly, it is not even necessary to forbid triangles to obtain
this result, just forbidden 4-cycles are sufficient.  Indeed, we can show the following stronger result.

\begin{theorem}\label{thm-algnobip}
Let $a\le b$ be positive integers and let $\GG$ be a proper minor-closed class of graphs.
There exists a polynomial-time algorithm taking as an input a graph $G\in\GG$ not containing $K_{a,b}$ as a subgraph
and returning an integer $c$ such that $\chi(G)\le c \le\chi(G)+a+4$.
\end{theorem}

Let us remark that the multiplicative $2$-approximation algorithm of Demaine et al.~\cite{demaine2005algorithmic} can be combined with
the algorithms of Theorems~\ref{thm-alg}, \ref{thm-algnotri}, and \ref{thm-algnobip} by returning the minimum of their results.
E.g., if $H$ is a $t$-apex graph, then there is a polynomial-time algorithm coloring an $H$-minor-free graph $G$ using
at most $\min(2\chi(G),\chi(G)+t+3)\le \beta\chi(G)+(2-\beta)(t+3)$ colors, for any $\beta$ such that $1\le \beta\le 2$;
the combined multiplicative-additive non-approximability bounds of Theorems~\ref{thm-nonapprox} and \ref{thm-k4free} are also
of interest in this context.

In Section~\ref{sec-constr}, we present a graph construction which we exploit to obtain the non-approximability results
in Section~\ref{sec-nonapprox}.  The approximation algorithms are presented in Section~\ref{sec-approx}.

\section{Tree-like product of graphs}\label{sec-constr}

Let $G$ and $H$ be graphs, and let $|V(G)|=n$ and $V(H)=\{u_1,\ldots, u_k\}$.
Let $T_{n,k}$ denote the rooted tree of depth $k+1$ such that each vertex at depth at most $k$ has
precisely $n$ children (the \emph{depth} of the tree is the number of vertices of a longest path
starting with its root, and the depth of a vertex $x$ is the number of vertices of the path from the
root to $x$; i.e., the root has depth $1$).  For each non-leaf vertex $x\in V(T_{n,k})$, let $G_x$ be a distinct copy of the graph $G$
and let $\theta_x$ be a bijection from $V(G_x)$ to the children of $x$ in $T_{n,k}$.  If $v\in V(G_x)$, $y$ is a non-leaf
vertex of the subtree of $T_{n,k}$ rooted in $\theta_x(v)$, and $z\in V(G_y)$, then we say that $v$ is a \emph{progenitor} of $z$.
The \emph{level} of $v$ is defined to be the depth of $x$ in $T_{n,k}$.  Note that a vertex at level $j$ has exactly one
progenitor at level $i$ for all positive $i<j$.
The graph $T(G,H)$ is obtained from the disjoint union of the graphs $G_x$ for non-leaf vertices $x\in V(T_{n,k})$
by, for each edge $u_iu_j\in E(H)$ with $i<j$, adding all edges from vertices of $T(G,H)$ at level $j$ to their
progenitors at level $i$.
Note that the graph $T(G,H)$ depends on the ordering of the vertices of $H$, which we consider to be fixed arbitrarily.

\begin{lemma}\label{lemma-propgen}
Let $G$ and $H$ be graphs with $|V(G)|=n\ge 2$ and $V(H)=\{u_1,\ldots, u_k\}$.  Let $q\le k$ be the maximum integer
such that $\{u_{k-q+1},\ldots,u_k\}$ is an independent set in $H$.  The graph $T(G,H)$ has $O(n^k)$ vertices
and $\omega(T(G,H))=\omega(G)+\omega(H)-1$.  Furthermore, if $F$ is a minor of $T(G,H)$ and $F$ is $(k-q+1)$-connected,
then there exists a set $X\subseteq V(F)$ of size at most $k-q$ such that $F-X$ is a minor of $G$.
\end{lemma}
\begin{proof}
The tree $T_{n,k}$ has $1+n+n^2+\ldots+n^{k-1}\le 2n^{k-1}$ non-leaf vertices, and thus $|V(T(G,H))|\le 2n^k$.

Consider a clique $K$ in $T(G,H)$, and let $v$ be a vertex of $K$ of largest level.  Let $x$ be the vertex of $T_{n,k}$
such that $v\in V(G_x)$.  Note that all vertices of $K\setminus V(G_x)$ are progenitors of $v$, and the vertices of $H$
corresponding to their levels are pairwise adjacent.  Consequently, $|K\setminus V(G_x)|\le \omega(H)-1$
and $|K\cap V(G_x)|\le \omega(G)$.  Therefore, each clique in $T(G,H)$ has size at most $\omega(G)+\omega(H)-1$.
A converse argument shows that cliques $K_G$ in $G$ and $K_H$ in $H$
give rise to a clique in $T(G,H)$ of size $|K_G|+|K_H|-1$, implying that $\omega(T(G,H))=\omega(G)+\omega(H)-1$.

For $i=0,\ldots, k-q$, let $G_i$ denote the graph obtained from $G$ by adding $i$ universal vertices.
Observe that $T(G,H)$ is obtained from copies of $G_0$, \ldots, $G_{k-q}$ by clique-sums on cliques of size
at most $k-q$ (consisting of the progenitors whose level is most $k-q$).  Hence, each $(k-q+1)$-connected minor $F$
of $T(G,H)$ is a minor of one of $G_0$, \ldots, $G_{k-q}$,
and thus a minor of $G$ can be obtained from $F$ by removing at most $k-q$ vertices.
\end{proof}

For an integer $p\ge 1$, the \emph{$p$-blowup} of a graph $H_0$ is the graph $H$ obtained from $H_0$ by replacing every
vertex $u$ by an independent set $S_u$ of $p$ vertices, and by adding all edges $zz'$ such that $z\in S_u$ and $z'\in S_{u'}$
for some $uu'\in E(H_0)$.  For the purposes of constructing the graph $T(G,H)$, we order the vertices of $H$ so that
for each $u\in V(H_0)$, the vertices of $S_u$ are consecutive in the ordering.
The \emph{strong $p$-blowup} is obtained from the $p$-blowup by making the sets $S_u$ into cliques
for each $u\in V(H_0)$.  For integers $a\ge b\ge 1$, an \emph{$(a:b)$-coloring} of $H_0$ is a function $\varphi$ that to each
vertex of $H_0$ assigns a subset of $\{1,\ldots, a\}$ of size $b$ such that $\varphi(u)\cap\varphi(v)=\emptyset$ for each edge
$uv$ of $H_0$.  The \emph{fractional chromatic number} $\chi_f(H_0)$ is the infimum of $\{a/b:\text{$H_0$ has an $(a:b)$-coloring}\}$.
Note that if $H$ is the strong $p$-blowup of a graph $H_0$, then a $c$-coloring of $H$ gives a $(c:p)$-coloring of $H_0$.
Consequently, we have the following.
\begin{observation}\label{obs-sbl}
If $H$ is the strong $p$-blowup of a graph $H_0$, then $\chi(H)\ge p\chi_f(H_0)$.
\end{observation}
We now state a key result concerning the chromatic number of the graph $T(G,H)$.

\begin{lemma}\label{lemma-propcolor}
Let $p,c\ge 1$ be integers and let $G$ be a graph.
Let $H_0$ be a graph such that $\chi(H_0)=\chi_f(H_0)=c$, and let $H$ be the $p$-blowup of $H_0$.
Then
$$\chi(T(G,H))\le c\chi(G)$$
and if $\chi(G)\ge p$, then
$$\chi(T(G,H))\ge cp.$$
\end{lemma}
\begin{proof}
Let $V(H)=\{u_1,\ldots, u_k\}$, where $k=p|V(H_0)|$.
Note that $\chi(H)\le \chi(H_0)=c$.  Let $\varphi_H$ be a proper coloring of $H$ using $c$ colors.
Let $C_1$, \ldots, $C_c$ be pairwise disjoint sets of $\chi(G)$ colors.
For each non-leaf vertex $x$ of $T_{n,k}$ of depth $i$, color $G_x$ properly using the colors in $C_{\varphi_H(u_i)}$.
Observe that this gives a proper coloring of $T(G,H)$ using at most $c\chi(G)$ colors,
and thus $\chi(T(G,H))\le c\chi(G)$.

Suppose now that $\chi(G)\ge p$ and consider a proper coloring $\varphi$ of $T(G,H)$.
Let $P=x_1x_2\ldots x_{k+1}$ be a path in $T_{n,k}$ from its root $x_1$ to one of the leaves and let $\psi$ be a coloring
of $H$ constructed as follows.
Suppose that we already selected $x_1$, \ldots, $x_i$ for some $i\le k$.  Let $Z_i$ denote the set of progenitors
of level at least $i-p+1$ of the vertices of $G_{x_i}$.  Since $|Z_i|\le p-1$ and $\varphi$ uses at least $\chi(G)\ge p$ distinct
colors on $G_{x_i}$, there exists $v\in V(G_{x_i})$ such that $\varphi(v)$ is different from the colors of all vertices of $Z_i$.
We define $x_{i+1}=\theta_{x_i}(v)$ be the child of $x_i$ in $T_{n,k}$ corresponding to $v$, and set $\psi(u_i)=\varphi(v)$.

Note that $\psi$ is a proper coloring of $H$ such that for each $u\in V(H_0)$, $\psi$ assigns vertices in $S_u$
pairwise distinct colors.  Consequently, $\psi$ is a proper coloring of the strong $p$-blowup of $H_0$, and
thus $\psi$ (and $\varphi$) uses at least $p\chi_f(H_0)=cp$ distinct colors by Observation~\ref{obs-sbl}.  We conclude that $\chi(T(G,H))\ge cp$.
\end{proof}

For positive integers $p$ and $k$, let $K_{k\times p}$ denote the $p$-blowup of $K_k$,
i.e., the complete $k$-partite graph with parts of size $p$.  Let us summarize the results of this section
in the special case of the graph $T(G_0,K_{k\times 4})$ with $G_0$ planar.

\begin{corollary}\label{cor-prop}
Let $G_0$ be a planar graph with $n$ vertices and let $k_0$ be a positive integer.  Let $G=T(G_0,K_{k_0\times 4})$.
The graph $G$ has $O(n^{4k_0})$ vertices. If $G_0$ is $3$-colorable, then $\chi(G)\le 3k_0$, and otherwise $\chi(G)\ge 4k_0$.
Furthermore, every $(4k_0-3)$-connected graph appearing as a minor in $G$ is $(4k_0-4)$-apex.
\end{corollary}
\begin{proof}
Note that $\chi(K_{k_0})=\chi_f(K_{k_0})=k_0$, $|V(K_{k_0\times 4})|=4k_0$ and the last $4$ vertices of $K_{k_0\times 4}$ form an independent set.
The claims follow from Lemma~\ref{lemma-propgen} (with $H=K_{k_0\times 4}$, $k=4k_0$ and $q=4$, using the fact that every minor of a planar graph is planar)
and Lemma~\ref{lemma-propcolor} (with $H_0=K_{k_0}$, $p=4$ and $c=k_0$).
\end{proof}

\section{Non-approximability}\label{sec-nonapprox}

The main non-approximability result is a simple consequence of Corollary~\ref{cor-prop} and NP-hardness
of testing $3$-colorability of planar graphs.

\begin{proof}[Proof of Theorem~\ref{thm-nonapprox}]
Suppose for a contradiction that there exists such a polynomial-time algorithm $\Aa$,
taking as an input an $F$-minor-free graph $G$ and returning an integer $c$ such that $\chi(G)\le c<\beta\chi(G)+(4-3\beta)k_0$.

Let $G_0$ be a planar graph, and let $G=T(G_0,K_{k_0\times 4})$.  By Corollary~\ref{cor-prop},
the size of $G$ is polynomial in the size of $G_0$ and $G$ is $F$-minor-free.
Furthermore, if $G_0$ is $3$-colorable, then $\chi(G)\le 3k_0$, and otherwise $\chi(G)\ge 4k_0$.
Hence, if $G_0$ is $3$-colorable, then the value returned by the algorithm $\Aa$
applied to $G$ is less than $\beta\chi(G)+(4-3\beta)k_0\le 4k_0$, and if $G_0$ is not $3$-colorable, then the value returned is
at least $\chi(G)\ge 4k_0$.  This gives a polynomial-time algorithm to decide whether $G_0$ is $3$-colorable.

However, it is NP-hard to decide whether a planar graph is $3$-colorable~\cite{garey1979computers},
which gives a contradiction unless $\text{P}=\text{NP}$.
\end{proof}

Note that the graphs $T(G_0,K_{k_0\times 4})$ used in the proof of Theorem~\ref{thm-nonapprox} have large cliques
(of size greater than $k_0$).  This turns out not to be essential---we can prove somewhat weaker non-approximability result
even for graphs with clique number $3$.
To do so, we need to apply the construction with both $G$ and $H_0$ being triangle-free.  A minor issue is that testing $3$-colorability
of triangle-free planar graphs is trivial by Gr\"otzsch' theorem~\cite{grotzsch1959}.  However, this can be easily worked around.
\begin{lemma}\label{lemma-npcnotria}
Let $\GG$ denote the class of graphs such that all their $3$-connected minors with at least $12$ vertices are planar.
The problem of deciding whether a triangle-free graph $G\in\GG$ is $3$-colorable is NP-hard.
\end{lemma}
\begin{proof}
Let $R_0$ be the Gr\"otzsch graph ($R_0$ is a triangle-free graph with $11$ vertices and chromatic number $4$,
and all its proper subgraphs are $3$-colorable).
Let $R$ be a graph obtained from $R_0$ by removing any edge $uv$.  Note that $R$ is $3$-colorable and the vertices
$u$ and $v$ have the same color in every $3$-coloring.

Let $G_1$ be a planar graph.  Let $G_2$ be obtained from $G_1$ by replacing each edge $xy$ of $G_1$
by a copy of $R$ whose vertex $u$ is identified with $x$ and an edge is added between $v$ and $y$
(i.e., $G_2$ is obtained from $G_1$ by a sequence of Haj\'os sums with copies of $R_0$).
Clearly, $G_2$ is triangle-free, it is $3$-colorable if and only if $G_1$ is $3$-colorable, and
$|V(G_2)|=|V(G_1)|+10|E(G_1)|$.
Furthermore, $G_2$ is obtained from a planar graph by clique-sums with $R_0$ on cliques of size two,
and thus every $3$-connected minor of $G_2$ is either planar or a minor of $R_0$
(and thus has at most $11$ vertices).  Hence, $G_2$ belongs to $\GG$.

Since testing $3$-colorability of planar graphs is NP-hard, it follows that testing $3$-colorability of triangle-free graphs from $\GG$
is NP-hard.
\end{proof}

We also need a graph $H_0$ which is triangle-free and its fractional chromatic number is large and equal to its ordinary
chromatic number.  We are not aware of such a construction being previously studied; in Appendix, we prove the following.

\begin{lemma}\label{lemma-constr}
For every positive integer $m$, there exists a triangle-free graph $H_m$ with $O(m^4\log^2 m)$ vertices
such that $\chi(H_m)=\chi_f(H_m)=m$.
\end{lemma}

Theorem~\ref{thm-k4free} now follows in the same way as Theorem~\ref{thm-nonapprox}, using the graphs from Lemma~\ref{lemma-constr}
instead of cliques.

\begin{proof}[Proof of Theorem~\ref{thm-k4free}]
Suppose for a contradiction that there exists such a polynomial-time algorithm $\Aa$,
taking as an input an $F$-minor-free graph $G$ with $\omega(G)\le 3$
and returning an integer $c$ such that $\chi(G)\le c<\beta\chi(G)+d$.
Recall that $m=\lceil d/(4-3\beta)\rceil$.
Let $H_m$ be the graph from Lemma~\ref{lemma-constr}. Let $k_0=|V(H_m)|$ and
let $H$ be the $4$-blowup of $H_m$.  Let $\GG$ be the class of graphs such that all their $3$-connected minors with at least $12$
vertices are planar.

Consider a triangle-free graph $G_0\in \GG$, and let $G=T(G_0,H)$.  By Lemma~\ref{lemma-propgen},
the size of $G$ is polynomial in the size of $G_0$.  Consider any $(4k_0-1)$-connected minor $F'$ of $G$.
By Lemma~\ref{lemma-propgen}, there exists a set $X$ of size at most $4k_0-4$ such that $F'-X$ is a minor of $G_0$.
Since $F'-X$ is $3$-connected and $G_0\in \GG$, we conclude that either $|V(F')|\le |X|+11\le 4k_0+7$ or $F'-X$ is planar.
Consequently, $F'\neq F$, and thus $G$ does not contain $F$ as a minor.
Furthermore, $\omega(G)\le\omega(G_0)+\omega(H)-1\le 3$.

Recall that $\chi(H_m)=\chi_f(H_m)=m$.
By Lemma~\ref{lemma-propcolor}, if $G_0$ is $3$-colorable, then $\chi(G)\le 3m$, and otherwise $\chi(G)\ge 4m$.
Hence, if $G_0$ is $3$-colorable, then the value returned by the algorithm $\Aa$
applied to $G$ is less than $\beta\chi(G)+d\le 3\beta m+d\le 4m$, and if $G_0$ is not $3$-colorable, then the value returned is
at least $\chi(G)\ge 4m$.  This gives a polynomial-time algorithm to decide whether $G_0$ is $3$-colorable,
in contradiction to Lemma~\ref{lemma-npcnotria} unless $\text{P}=\text{NP}$.
\end{proof}

\section{Approximation algorithms}\label{sec-approx}

Let us now turn our attention to the additive approximation algorithms.
The algorithms we present use ideas similar to the ones of the $2$-approximation algorithm of Demaine et al.~\cite{demaine2005algorithmic}
and of the additive approximation algorithms of Kawarabayashi et al.~\cite{kawarabayashi2009additive} and Demaine et al.~\cite{demaine2009approximation}.
We find a partition of the vertex set of the input graph $G$ into parts $L$ and $C$ such that $G[L]$ has bounded tree-width
(and thus its chromatic number can be determined exactly) and $G[C]$ has bounded chromatic number, and color the parts
using disjoint sets of colors.
The existence of such a decomposition is proved using the minor structure theorem~\cite{robertson2003graph}, in the variant
limiting the way apex vertices attach to the surface part of the decomposition.  The proof of this stronger version can be
found in~\cite{apex}.  Let us now introduce definitions necessary to state this variant of the structure theorem.

A \emph{tree decomposition} of a graph $G$ is a pair $(T,\beta)$, where $T$ is a tree and $\beta$ is a function
assigning to each vertex of $T$ a subset of $V(G)$, such that for each $uv\in E(G)$ there exists $z\in V(T)$ with
$\{u,v\}\subseteq \beta(z)$, and such that for each $v\in V(G)$, the set $\{z\in V(T):v\in\beta(z)\}$ induces a non-empty
connected subtree of $T$.  The \emph{width} of the tree decomposition is $\max\{|\beta(z)|:z\in V(T)\}-1$,
and the \emph{tree-width} of a graph is the minimum of the widths of its tree decompositions.

The decomposition is \emph{rooted} if $T$ is rooted.
For a rooted tree decomposition $(T,\beta)$ and a vertex $z$ of $T$ distinct from the root,
if $w$ is the parent of $z$ in $T$, we write $\beta\uparrow z\colonequals \beta(z)\cap\beta(w)$ and
$\beta\downarrow z\colonequals \beta(z)\setminus \beta(w)$.  If $z$ is the root of $T$, then
$\beta\uparrow z\colonequals\emptyset$ and $\beta\downarrow z\colonequals\beta(z)$.  
The \emph{torso expansion} of a graph $G$ with respect to its rooted tree decomposition $(T,\beta)$
is the graph obtained from $G$ by adding edges of cliques on $\beta\uparrow z$ for all $z\in V(T)$.

A \emph{path decomposition} is a tree decomposition $(T,\beta)$ where $T$ is a path.
A \emph{vortex with boundary sequence $v_1$, \ldots, $v_s$ and depth $d$} is a graph with a path decomposition
$(p_1p_2\ldots p_s,\beta)$ such that $|\beta(p_i)|\le d+1$ and $v_i\in \beta(p_i)$ for $i=1,\ldots,s$.
An embedding of a graph in a surface is \emph{$2$-cell} if each face of the embedding is homeomorphic to an open disk.
\begin{theorem}\label{thm-struct}
For every non-negative integer $t$ and a $t$-apex graph $H$, there exists a constant $a_H$ such that the following holds.
For every $H$-minor-free graph $G$, there exists a rooted tree decomposition $(T,\beta)$ of $G$ with the following properties.
Let $G'$ denote the torso expansion of $G$ with respect to $(T,\beta)$.
For every $v\in V(T)$, there exists a set $A_v\subseteq \beta(v)$ of size at most $a_H$ with $\beta\uparrow v\subseteq A_v$,
a set $A'_v\subseteq A_v$ of size at most $t-1$, and
subgraphs $G_v$, $G_{v,1}$, \ldots, $G_{v,m}$ of $G'[\beta(v)\setminus A_v]$ for some $m\le a_H$ such that
\begin{itemize}
\item[(a)] $G'[\beta(v)\setminus A_v]=G_v\cup G_{v,1}\cup \ldots\cup G_{v,m}$, and for $1\le i<j\le m$,
the graphs $G_{v,i}$ and $G_{v,j}$ are vertex-disjoint and $G'$ contains no edges between $V(G_{v,i})$ and $V(G_{v,j})$,
\item[(b)] the graph $G_v$ is $2$-cell embedded in a surface $\Sigma_v$ in that $H$ cannot be drawn,
\item[(c)] for $1\le i\le m$, $G_{v,i}$ is a vortex of depth $a_H$ intersecting $G_v$ only in its boundary sequence,
and this sequence appears in order in the boundary of a face $f_{v,i}$ of $G_v$, and $f_{v,i}\neq f_{v,j}$ for $1\le i<j\le m$,
\item[(d)] vertices of $G_v$ have no neighbors in $A_v\setminus A'_v$, and
\item[(e)] if $w$ is a child of $v$ in $T$ and $\beta(w)\cap V(G_v)\neq \emptyset$, then 
$\beta\uparrow w\subseteq V(G_v)\cup A'_v$.
\end{itemize}
Furthermore, the tree decomposition and the sets and subgraphs as described can be found in polynomial time.
\end{theorem}
Informally, the graph $G$ is a clique-sum of the graphs $G'[\beta(v)]$ for $v\in V(T)$, and $\beta(v)$ contains a bounded-size set $A_v$
of apex vertices such that $G'[\beta(v)]-A_v$ can be embedded in $\Sigma_v$ up to a bounded number of vortices of bounded depth.
Furthermore, at most $t-1$ of the apex vertices (forming the set $A'_v$) can have neighbors in the part $G_v$ of $G'[\beta(v)]-A_v$
drawn in $\Sigma_v$, or in the other bags of the decomposition that intersect $G_v$.  Note that it is also possible that
$\Sigma_v$ is the null surface, and consequently $A_v=\beta(v)$.

We need the following observation on graphs embedded up to vortices.
\begin{lemma}\label{lemma-remove}
Let $\Sigma$ be a surface of Euler genus $g$ and let $a$ be a non-negative integer.  Let $G$ be a graph and let $G_0$, $G_1$, \ldots, $G_m$
be its subgraphs such that $G=G_0\cup G_1\cup \ldots \cup G_m$, the subgraphs $G_1$, \ldots, $G_m$
are pairwise vertex-disjoint and $G$ contains no edges between them, $G_0$ is $2$-cell embedded in $\Sigma$,
and there exist pairwise distinct faces $f_1$, \ldots, $f_m$ of this embedding such that for $1\le i\le m$,
$G_i$ intersects $G_0$ only in a set of vertices contained in the boundary of $f_i$.
If the graphs $G_1$, \ldots, $G_m$ have tree-width at most $a$, then there exists a subset $L_0$ of vertices of $G_0$
such that $G_0-L_0$ is planar and the graph $G_0[L_0]\cup G_1\cup\ldots\cup G_m$ has tree-width at most $26g+9m+a$.
\end{lemma}
\begin{proof}
If $\Sigma$ is the sphere, then we can set $L_0=\emptyset$; hence, we can assume that $g\ge 1$.
Let $G'_0$ be the graph obtained from $G_0$ by, for $1\le i\le m$, adding a new vertex $r_i$ drawn inside $f_i$
and joined by edges to all vertices of $V(G_0)\cap V(G_i)$.  Note that $G'_0$ has a $2$-cell embedding in $\Sigma$ extending the embedding of $G$.
Let $F$ be a subgraph of $G'_0$ such that the embedding of $F$ in $\Sigma$ inherited from the embedding of $G'_0$
is $2$-cell and $|V(F)|+|E(F)|$ is minimum.  Then $F$ has only one face, since otherwise it is possible to remove an edge separating distinct faces of $F$,
and $F$ has minimum degree at least two, since otherwise we can remove a vertex of degree at most $1$ from $F$.
By generalized Euler's formula, we have $|E(F)|=|V(F)|+g-1$, and thus $F$ contains at most $2(g-1)$ vertices of degree
greater than two.  By considering the graph obtained from $F$ by suppressing vertices of degree two, we see that
$F$ is either a cycle (if $g=1$) or a subdivision of a graph with at most $3(g-1)$ edges.

Let $M_0$ be the set of vertices of $F$ of degree at least three and their neighbors.
We claim that each vertex of $V(F)\setminus M_0$ is adjacent in $G'_0$ to only two vertices of $V(F)$.
Indeed, suppose that a vertex $x\in V(F)\setminus M_0$ has at least three neighbors in $G'_0$ belonging to $V(F)$.
Let $w$ and $y$ be the neighbors of $x$ in $F$, and let $z$ be a vertex distinct from $w$ and $y$ adjacent to $x$ in $G'_0$.
The graph $F+xz$ has two faces, and by symmetry, we can assume that the edge $xy$ separates them.  Since $x\not\in M_0$,
the vertex $y$ has degree two in $F$, and thus the embedding of $F+xz-y$ is $2$-cell, contradicting the minimality of $F$.

Let $M_1$ be the set of vertices of $F$ at distance at most $4$ from a vertex of degree greater than two.  Note that $|M_1|\le 26g$.
For $1\le i\le m$, let $N_i$ denote the set of vertices of $F-M_1$ that are in $G'_0$ adjacent to a vertex of $V(G_i)\cap V(G_0)\setminus V(F)$.
We claim that $|N_i|\le 9$.  Indeed, suppose for a contradiction that $|N_i|\ge 10$ and consider a path $w_4w_3w_2w_1xy_1y_2y_3y_4$ in $F$ such that
$x$ belongs to $N_i$ (vertices $x$, $w_1$, \ldots, $w_4$, $y_1$, \ldots, $y_4$ have degree two in $F$, since $x\not\in M_1$).
If $r_i\in V(F)$, then let $Q$ be a path in $G'_0$ of length two between $x$ and $r_i$ through a vertex of $V(G_i)\cap V(G_0)\setminus V(F)$; note that
$r_i\not\in\{w_1,w_2,y_1,y_2\}$, since $r_i$ has at least $10$ neighbors in $V(F)$ and belongs to $M_0$ by the previous paragraph
and $x\not\in M_1$.  If $r_i\not\in V(F)$, then there exists a vertex $z\in N_i\setminus \{x,w_1,\ldots, w_4,y_1,\ldots,y_4\}$;
we let $Q$ be a path of length at most four between $x$ and $z$ passing only through their neighbors in $V(G_i)\cap V(G_0)\setminus V(F)$ and possibly through $r_i$.
By symmetry, we can assume that the edge $xy_1$ separates the two faces of $F+Q$, and the graph
$F+Q-\{y_1,y_2\}$ if $r_i\in V(F)$ or $F+Q-\{y_1,y_2,y_3,y_4\}$ if $r_i\not\in V(F)$ contradicts the minimality of $F$.

Let $M=(M_1\cup \bigcup_{i=1}^m N_i)\cap V(G_0)$; we have $|M|\le 26g+9m$.  Let $L_0=V(F)\cap V(G_0)$.
Clearly, $G_0-L_0\subseteq G'_0-V(F)$ is planar.  Let $G'$ be the graph obtained from $G_1\cup\ldots\cup G_m$ by adding
vertices of $M$ as universal ones, adjacent to all other vertices of $G'$.  The tree-width of $G'$ is at most $a+|M|\le 26g+9m+a$.
Note that each vertex of $L_0\setminus V(G')$ has degree two in $G_0[L_0]\cup G_1\cup\ldots\cup G_m$,
and thus $G_0[L_0]\cup G_1\cup\ldots\cup G_m$ is a subgraph of a subdivision of $G'$.  We conclude that
the tree-width of $G_0[L_0]\cup G_1\cup\ldots\cup G_m$ is also at most $26g+9m+a$.
\end{proof}
Note that the set $L_0$ can be found in polynomial time.  For the clarity of presentation of the proof we selected $F$ with $|V(F)|+|E(F)|$
minimum; however, it is sufficient to start with an arbitrary inclusionwise-minimal subgraph with exactly one $2$-cell face
(obtained by repeatedly removing edges that separate distinct faces and vertices of degree at most $1$) and
repeatedly perform the reductions described in the proof until each vertex of $V(F)\setminus M_0$ is adjacent in $G'_0$ to only
two vertices of $V(F)$ and until $|N_i|\le 9$ for $1\le i\le m$.

For positive integers $t$ and $a$, we say that a rooted tree decomposition $(T,\beta)$ of a graph $G$ is \emph{$(t,a)$-restricted}
if for each vertex $v$ of $T$, the subgraph of the torso expansion of $G$ induced by $\beta\downarrow v$
is planar, $|\beta\uparrow v|\le t-1$, and each vertex of $\beta\downarrow v$ has at most $a-1$ neighbors in $G$ that belong to
$\beta\uparrow v$.  Using the decomposition from Theorem~\ref{thm-struct}, we now partition the considered graph
into a part of bounded tree-width and a $(t,t)$-restricted part.

\begin{theorem}\label{thm-decomp}
For every positive integer $t$ and a $t$-apex graph $H$, there exists a constant $c_H$ with the following property.
The vertex set of any $H$-minor-free graph $G$ can be partitioned in polynomial time into two parts $L$ and $C$ such that
$G[L]$ has tree-width at most $c_H$ and $G[C]$ has a $(t,t)$-restricted rooted tree decomposition.
Additionally, for any such graph $H$ and positive integers $a\le b$, there exists a constant $c_{H,a,b}$ such that
if $G$ is $H$-minor-free and does not contain $K_{a,b}$ as a subgraph, then $L$ and $C$ can be chosen so that $G[L]$ has tree-width at
most $c_{H,a,b}$ and $G[C]$ has a $(t,a)$-restricted rooted tree decomposition.
\end{theorem}
\begin{proof}
Let $a_H$ be the constant from Theorem~\ref{thm-struct} for $H$, and let $g$ be the maximum Euler genus of a surface in that
$H$ cannot be embedded.  Let $c_H=26g+11a_H$ and $c_{H,a,b}=c_H+\binom{t-1}{a}(b-1)$.

Since $G$ is $H$-minor-free, we can in polynomial time find its rooted tree decomposition $(T,\beta)$, its torso expansion $G'$,
and for each $v\in V(T)$, find $A_v$, $A'_v$, $G_v$, $G_{v,1}$, \ldots, $G_{v,m}$, and $\Sigma_v$ as described
in Theorem~\ref{thm-struct}.  Let $L'_v$ be the set of vertices obtained by applying Lemma~\ref{lemma-remove}
to $G_v$, $G_{v,1}$, \ldots, $G_{v,m}$; i.e., $G_v-L'_v$ is planar and the graph
$G_v[L'_v]\cup G_{v,1}\cup\ldots\cup G_{v,m}$ has tree-width at most $26g+10a_H$.
When considering the case that $G$ does not contain $K_{a,b}$ as a subgraph, let $S_v$ be the set of vertices of $G_v$ that have at least $a$ neighbors
in $G$ belonging to $A_v$ (and thus to $A'_v$); otherwise, let $S_v=\emptyset$.
Since there are at most $\binom{t-1}{a}$ ways how to choose a set of $a$ neighbors in $A'_v$ and no $b$ vertices can have the
same set of $a$ neighbors in $A'_v$, we have $|S_v|\le \binom{t-1}{a}(b-1)$.
Let $L_v=(A_v\setminus \beta\uparrow v)\cup V(G_{v,1}\cup \ldots\cup G_{v,m})\cup L'_v\cup S_v$.

We define $L=\bigcup_{v\in V(T)} L_v$.  Note that $L\cap \beta(v)\subseteq L_v\cup (\beta\uparrow v)\subseteq L_v\cup A_v$.
Consequently, $G'[L\cap \beta(v)]$ is obtained from $G_v[L'_v]\cup G_{v,1}\cup\ldots\cup G_{v,m}$ by adding $S_v$ and some of the
vertices of $A_v$, and consequently $G'[L\cap \beta(v)]$ has tree-width at most $c_{H,a,b}$ when considering the case that $G$ does not contain $K_{a,b}$ as a subgraph
and at most $c_H$ otherwise.  The graph $G[L]$ is a clique-sum of the graphs
$G'[L\cap\beta(v)]$ for $v\in V(T)$, and thus the tree-width of $G[L]$ is also at most $c_{H,a,b}$ or $c_H$.

Let $C=V(G)\setminus L$, and consider the graph $G[C]$.  For $v\in V(T)$, let $\beta'(v)=\beta(v)\cap C$.
Then $(T,\beta')$ is a rooted tree decomposition of $G[C]$ such that for every $v\in V(T)$,
the graph $G'[\beta'\downarrow v]\subseteq G_v-L'_v$ is planar, all vertices of $\beta'\uparrow v$ adjacent in $G$
to a vertex of $\beta'\downarrow v$ belong to $A'_v$ (and thus there are at most $t-1$ such vertices), and
when considering the case that $G$ does not contain $K_{a,b}$ as a subgraph, each vertex of $\beta'\downarrow v$ has at most $a-1$
neighbors in $G$ belonging to $\beta'\uparrow v$.

Note that $\beta'\uparrow v$ can contain vertices not belonging
to $A'_v$, and thus $\beta'\uparrow v$ can have size larger than $t-1$, and
the tree decomposition $(T,\beta')$ is not necessarily $(t,t)$-restricted.
However, by the condition (e) from the statement of Theorem~\ref{thm-struct},
the vertices of $(\beta'\uparrow v)\setminus A'_v$ can only be contained in the
bags of descendants of $v$ which are disjoint from $V(G_v)$, and thus we can
fix up this issue as follows.

If $w$ is a child of $v$ and $\beta'(w)\cap V(G_v)=\emptyset$, we say that the edge $vw$ is \emph{skippable}; note that
in that case $\beta'\uparrow w\subseteq \beta'\uparrow v$.  For each vertex $w$ of $T$, let $f(w)$ be the nearest
ancestor of $w$ such that the first edge on the path from $f(w)$ to $w$ in $T$ is not skippable.
Let $T'$ be the rooted tree with vertex set $V(T)$ where the parent of each vertex $w$ is $f(w)$.  Observe that $(T',\beta')$
is a tree decomposition of $G[C]$.  Furthermore, denoting by $z$ the child of $f(w)$ on
the path from $f(w)$ to $w$ in $T$, note that if a vertex $x\in \beta'\uparrow f(w)$ is contained in $\beta'(w)$, then
$x\in \beta(z)$, and since the edge $f(w)z$ is not skippable, the condition (e) from the statement of Theorem~\ref{thm-struct}
implies that $x\in A'_{f(w)}$.

Hence, letting $\beta''(v)=(\beta'\downarrow v)\cup (A'_v\cap C)$
for each vertex $v$ of $T'$, we conclude that $(T',\beta'')$ is a rooted tree decomposition of $G[C]$ which
is $(t,t)$-restricted, and when considering the case that $G$ does not contain $K_{a,b}$ as a subgraph, the decomposition is $(t,a)$-restricted.
\end{proof}

Let us now consider the chromatic number of graphs with a $(t,a)$-restricted tree decomposition.

\begin{lemma}\label{lemma-color}
Let $a$ and $t$ be positive integers.
Let $G$ be a graph with a $(t,a)$-restricted rooted tree decomposition $(T,\beta)$.
The chromatic number of $G$ is at most $\min(t+3,a+4)$.
Additionally, if $G$ is triangle-free, then the chromatic number of $G$ is at most $\lceil (13t+172)/14\rceil$. 
\end{lemma}
\begin{proof}
We can color $G$ using $t+3$ colors, starting from the root of the tree decomposition, as follows.
Suppose that we are considering a vertex $v\in V(T)$ such that $\beta\uparrow v$ is already colored.
Since $|\beta\uparrow v|\le t-1$, this leaves at least $4$ other colors to be used on
$G[\beta\downarrow v]$.  Hence, we can extend the coloring to $G[\beta\downarrow v]$ by the Four Color Theorem.

We can also color $G$ using $a+4$ colors, starting from the root of the tree decomposition, as follows.
For each vertex $x$ of $\beta\downarrow v$, at most $a-1$ colors are used on its neighbors in
$\beta\uparrow v$, leaving $x$ with at least $5$ available colors not appearing on its neighbors.  Since
$G[\beta\downarrow v]$ is planar, we can color it from these lists of size at least $5$ using the result
of Thomassen~\cite{thomassen1994}, again extending the coloring to $G[\beta\downarrow v]$.

Finally, suppose that $G$ is triangle-free.  Let $G'$ be the torso expansion of $G$ with respect to $(T,\beta)$, and let $c=\lceil (13t+172)/14\rceil$. 
We again color $G$ starting from the root of the tree decomposition using at most
$c$ colors.  Additionally, we choose the coloring so that the following invariant is satisfied: ($\star$) for each vertex $w$ of $T$
and for each independent set $I$ in $G[\beta(w)]$ such that $I\cap \beta\downarrow w$ induces a clique in $G'$,
at most $c-6$ distinct colors are used on $I$.

Let $v$ be a vertex of $T$.  Suppose we have already colored
$\beta\uparrow v$, and we want to extend the coloring to $\beta\downarrow v$.
Note that the choice of this coloring may only affect the validity of the invariant ($\star$) at $v$ and
at descendants of $v$ in $T$.  Consider any descendant $w$ of $T$.  Coloring $\beta\downarrow v$
can only assign color to vertices of $\beta\uparrow w$, and since
$G'$ is the torso expansion of $G$, the set $\beta\uparrow w\cap \beta\downarrow v$ induces a clique in $G'$.
Consequently, the validity of ($\star$) at $v$ implies the validity at $w$ (until more vertices of $G$ are assigned colors),
and thus when choosing the coloring of $\beta\downarrow v$, we only need to ensure that ($\star$) holds at $v$.

The graph $G'[\beta\downarrow v]$ is planar, and thus it is $5$-degenerate; i.e., there exists
an ordering of its vertices such that each vertex is preceded by at most $5$ of its neighbors.
Let us color the vertices of $\beta\downarrow v$ according to this ordering, always preserving the validity
of ($\star$) at $v$.  Suppose that we are choosing a color for a vertex $x\in V(\beta\downarrow v)$.  Let $P_x$
consist of the neighbors of $x$ in $G'$ belonging to $\beta\downarrow v$ that precede it in the ordering;
we have $|P_x|\le 5$.  Note that all cliques in $G'[\beta\downarrow v]$ containing $x$ and with all other vertices already colored
are subsets of $P_x\cup \{x\}$.  Let $Q_x=P_x\cup \beta\uparrow v$; we have $|Q_x|\le t+4$.

Let $N_x$ consist of vertices of $Q_x$ that are adjacent to $x$ in $G$.
We say that a color $a$ is forbidden at $x$ if there exists an independent set
$A_a\subseteq Q_x\setminus N_x$ of $G$ such that $A_a\cap P_x$ is a clique in $G'$ and $c-6$ colors distinct from $a$ appear on $A_a$.
Observe that assigning $x$ a color which neither appears on $N_x$ nor is forbidden results
in a proper coloring that preserves the invariant ($\star$) at $v$.

Suppose first that no color is forbidden at $x$.  Since $G$ is triangle-free, $N_x\setminus P_x$ is an independent set in
$G[\beta\uparrow v]$, and by ($\star$), at most $c-6$ colors appear on $N_x\setminus P_x$.  Since $|P_x|\le 5$, it follows that
some color does not appear on $N_x$, as required.

Hence, we can assume that some color is forbidden at $x$, and thus there exists an independent
set $Z_1\subseteq Q_x\setminus N_x$ of size at least $c-6$ such that vertices of $Z_1$ are assigned pairwise distinct colors.
Since $|Q_x|\le t+4$, at most $t+4-(c-6)=t+10-c$ of these colors appear at least twice
on $Q_x$, and thus there exists a set $Z_2\subseteq Z_1$ of size at least $c-6-(t+10-c)=2c-t-16$
such that the color of each vertex of $Z_2$ appears exactly once on $Q_x$ (and thus does not appear on $N_x$).  Let $Z=Z_2\setminus P_x$;
we have $|Z|\ge 2c-t-21$.   We claim that not all colors appearing on $Z$ are forbidden; since
such colors do not appear on $N_x$, we can use them to color $x$.

For contradiction, assume that colors of all vertices of $Z$ are forbidden at $x$.  Let $Z=\{z_1,\ldots, z_m\}$
for some $m\ge 2c-t-21$, and for $a=1,\ldots,m$, let $a$ be the color of $z_a$.  Since $a$ is forbidden at $x$,
there exists an independent set $A_a\subseteq Q_x\setminus N_x$ such that $A_a\cap P_x$ is a clique in $G'$ and $c-6$ colors
distinct from $a$ appear on $A_a$.  Note that $A_a\cup \{z_a\}$ is not an independent set, as otherwise
this set contradicts the invariant ($\star$) at $v$.  Hence, we can choose a neighbor $f(a)$ of $z_a$ in $A_a$.
Since $Z$ is an independent set, we have $f(a)\not\in Z$.  Furthermore, we claim that the preimage in $f$ of each vertex
has size at most $6$: if say $f(z_1)=\ldots=f(z_7)=y$, then for $i=1,\ldots,7$, the vertex $z_i$ would have a neighbor $y$ in the independent set
$A_1$, and thus $z_1,\ldots,z_7\not\in A_1$; however, the only appearance of colors $1$, \ldots, $7$ in $Q_x$ is on the vertices
$z_1$, \ldots, $z_7$, and thus at most $c-7$ colors would appear on $A_1$.  We conclude that $|f(Z)|\ge |Z|/6$,
and thus $t+4\ge |Q_x|\ge |Z|+|f(Z)|\ge \tfrac{7}{6}|Z|\ge \tfrac{7}{6}(2c-t-21)\ge (6t+25)/6$.  This is a contradiction.
\end{proof}

Combining Lemma~\ref{lemma-color} with Theorem~\ref{thm-decomp} (coloring $G[L]$ using at most $\chi(G)$ colors in linear time~\cite{coltweasy},
using the fact that $G[L]$ has bounded tree-width, and coloring $G[C]$ using a disjoint set of colors),
we obtain Theorems~\ref{thm-alg}, \ref{thm-algnotri}, and \ref{thm-algnobip}.

\bibliographystyle{siam}
\bibliography{nonapprox}

\section*{Appendix}

We now present the proof of Lemma~\ref{lemma-constr}.
We use Talagrand's inequality in the following form (see~\cite[Talagrand's Inequality II]{molloy2013graph}).

\begin{theorem}\label{thm-talagrand}
Let $X$ be a non-negative random variable,
not identically $0$, which is determined by $n$ independent trials $T_1$, \ldots , $T_n$, and
satisfying the following for some $c, r > 0$:
\begin{enumerate}
\item Changing the outcome of any one trial can affect $X$ by at most $c$.
\item For any $s$, if $X\ge s$, then there is a set of at most $rs$ trials whose outcomes
certify that $X\ge s$.
\end{enumerate}
Then for any non-negative $t\le \text{E}[X]$,
$$\text{Pr}\bigl[|X-\text{E}[X]|>t+60c\sqrt{r\text{E}[X]}\bigr]\le 4e^{-\frac{t^2}{8c^2r\text{E}[X]}}.$$
\end{theorem}

We also use Chernoff's inequality (see~\cite[Chernoff Bound]{molloy2013graph}).
\begin{theorem}\label{thm-chernoff}
Let $X_1$, \ldots, $X_n$ be independent random variables, each taking value $1$ with probability $p$ and value $0$ otherwise.
For any non-negative $t\le np$,
$$\text{Pr}[|X-np|>t]<2e^{-\frac{t^2}{3np}}.$$
\end{theorem}

Let $m\ge 2$ be an integer.  Let $V$ be the disjoint union of $m$ sets $V_1$, \ldots, $V_m$ of vertices, each of size $n=\lceil 10^{13}m^3\log^2 m\rceil$.
Let $p=1/\sqrt{6(nm-2)}$.  Let $G_0$ be a random graph with vertex set $V$ in that each pair $uv$ forms an edge independently
with probability $p$.  We say that a set of edges of $G_0$ is \emph{triangular} if it is the union of edge sets
of pairwise edge-disjoint triangles of $G_0$, and a set is \emph{subtriangular} if it is a subset of a triangular set.
A triangle in $G_0$ is \emph{isolated} if it is edge-disjoint from all other triangles in $G_0$.
Let $G$ be a graph obtained from $G_0$ by removing edges of an inclusionwise-maximal triangular set as well all edges
with both ends in the same set $V_i$ for $i=1,\ldots, m$.  Clearly, $G$ is triangle-free and $\chi(G)\le m$.

\begin{lemma}\label{lemma-single}
Let $a$ be an integer such that $1\le a\le (1-1/m)n$.
Suppose $A$ is a subset of $V$ of size $n$ such that $|A\setminus V_1|=a$.
Let $S$ be the set of pairs with one end in $A\cap V_1$ and the other end in $A\setminus V_1$.
Then the probability that $S\cap E(G)=\emptyset$ is less than $6e^{-\tfrac{pn}{2592m}a}$.
\end{lemma}
\begin{proof}
Let $X_1=|S\cap E(G_0)|$ and let $X_2$ be the maximum size of a subtriangular (in $G_0$) subset of $S$.
If $X_1>X_2$, then not all edges of $S\cap E(G_0)$ can be contained in a triangular set in $G_0$, and
thus $\text{Pr}[S\cap E(G)=\emptyset]\le \text{Pr}[X_1\le X_2]$.
We have $\text{E}[X_1]=p|S|$, and by Theorem~\ref{thm-chernoff}, we have $\text{Pr}[X_1<p|S|/2]<2e^{-p|S|/12}$.

Let $X'_2$ be the number of edges of $S\cap E(G_0)$ that belong to triangles in $G_0$, and let $X''_2$
be the number of the edges of $S\cap E(G_0)$ that belong to isolated triangles.  Note that $X''_2\le X_2\le X'_2$.
Consider a pair $e\in S$.  Conditionally under the assumption that $e$ is an edge of $G_0$, the probability
that $e$ belongs to a triangle is at most $p^2(nm-2)=1/6$, and the probability that $e$ belongs to an isolated triangle
is $(nm-2)p^2((1-p)^3+3p(1-p)^2)^{nm-3}\ge (nm-2)p^2(1-3p^2)^{nm-3}\ge (nm-2)p^2(1-3p^2(nm-3))>1/12$.
Hence, $p|S|/12<\text{E}[X''_2]\le \text{E}[X_2]\le \text{E}[X'_2]\le p|S|/6$.

Note that for any pair $uv$, whether $uv$ is an edge of $G_0$ or not changes $X_2$ by at most $3$, and that
$|X_2|\ge s$ is certified by a set of at most $3s$ edges of $G_0$.  Hence, we can apply Theorem~\ref{thm-talagrand} for $X_2$
with $c=r=3$.  Furthermore, $\text{E}[X_2]\ge p|S|/12\ge pn/12>97200$, and thus $\text{E}[X_2]>180\sqrt{3\text{E}[X_2]}$.
Consequently,
$\text{Pr}[X_2\ge p|S|/2]\le \text{Pr}[X_2\ge 3\text{E}[X_2]]\le \text{Pr}[|X_2-\text{E}[X_2]|>\text{E}[X_2]+180\sqrt{3\text{E}[X_2]}]\le 4e^{-\text{E}[X_2]/216}\le 4e^{-p|S|/2592}$.
It follows that $\text{Pr}[X_1\le X_2]\le \text{Pr}[X_1<p|S|/2]+\text{Pr}[X_2\ge p|S|/2]\le 2e^{-p|S|/12}+4e^{-p|S|/2592}<6e^{-p|S|/2592}$.
Since $a\le (1-1/m)n$, we have $|S|=a(n-a)\ge na/m$.
Consequently, $\text{Pr}[S\cap E(G)=\emptyset]\le \text{Pr}[X_1\le X_2]<6e^{-p|S|/2592}\le 6e^{-\tfrac{pn}{2592m}a}$.
\end{proof}

\begin{lemma}\label{lemma-exin}
The probability that $G$ contains an independent set of size $n$ distinct from $V_1$, \ldots, $V_m$ is
less than $1$.
\end{lemma}
\begin{proof}
Each set of size $n$ intersects one of $V_1$, \ldots, $V_m$ in at least $n/m$ vertices.
Hence, it suffices to prove that for $i=1,\ldots,m$, the probability $p_i$ that $G$ contains an independent
set of size $n$ distinct from $V_i$ that intersects $V_i$ in at least $n/m$ vertices is less than $1/m$.
By symmetry, it suffices to show that $p_1<1/m$.

Consider an integer $a$ such that $1\le a\le (1-1/m)n$.  There are less than $(n^2m)^a$ sets $A\subseteq V$
of size $n$ such that $|A\setminus V_1|=a$ (one needs to select $a$ vertices of $|A\setminus V_1|$ from $V\setminus V_1$
and $a$ vertices of $V_1\setminus A$ from $V_1$).  By Lemma~\ref{lemma-single}, each such set has probability less than $6e^{-\tfrac{pn}{2592m}a}$
of being independent in $G$.  Note that $e^{-\tfrac{pn}{2592m}}n^2m<\tfrac{1}{12m}$.
Hence, $p_1<\sum_{a\ge 1} 6e^{-\tfrac{pn}{2592m}a}(n^2m)^a\le 6\sum_{a\ge 1} \bigl(\tfrac{1}{12m}\bigr)^a\le 1/m$.
\end{proof}

Lemma~\ref{lemma-exin} implies that with non-zero probability, $G$ has no large independent sets, giving a lower bound on
its chromatic number.

\begin{proof}[Proof of Lemma~\ref{lemma-constr}]
For $m=1$, the claim is trivial, hence assume that $m\ge 2$.
As we observed before, the graph $G$ constructed in this appendix is triangle-free and $\chi(G)\le m$.
On the other hand, Lemma~\ref{lemma-exin} implies that with non-zero probability, $G$ has no independent set of size $n+1$,
and thus $\chi(G)\ge \chi_f(G)\ge |V(G)|/\alpha(G)=m$.  Hence, we can set $H_m\colonequals G$.
\end{proof}

\end{document}